\documentclass[aps,pra,reprint,groupedaddress]{revtex4-1}

\usepackage{graphicx}
\usepackage{dcolumn}
\usepackage{bm}
\usepackage{amssymb}
\usepackage{amsmath}
\usepackage{amsthm}
\usepackage{color}

\newtheorem{theorem}{Theorem}
\newtheorem{lemma}{Lemma}

\begin{document}


\title{Efficient quantum repeater in perspectives of both entanglement concentration rate and LOCC complexity}

\author{Zhaofeng Su}
\affiliation{Centre for Quantum Software and Information, University of Technology Sydney, Ultimo NSW 2007, Australia.} 
\author{Ji Guan}
\affiliation{Centre for Quantum Software and Information, University of Technology Sydney, Ultimo NSW 2007, Australia.}

\author{Lvzhou Li}%
\affiliation{Institute of Computer Science Theory, School of Data and Computer Science, Sun Yat-Sen University, Guangzhou, 510006, China}

\date{\today}

\begin{abstract}
  Quantum entanglement is an indispensable resource for many significant quantum information processing tasks. However, because of the noise in quantum channels, it is difficult to distribute quantum entanglement over a long distance in practice. A solution for this challenge is the quantum repeater which can extend the distance of entanglement distribution. In this scheme, the time consumption of classical communication and local operations takes an important place in perspective of time efficiency. Motivated by this observation, we exploit the basic quantum repeater scheme in perspectives of not only the optimal rate of entanglement concentration but also the complexity of local operations and classical communication. Firstly, we consider the case where two two-qubit pure states are prepared. We construct a protocol with the optimal entanglement concentration rate and less consumption of local operations and classical communication. We also find a criteria for the projective measurements to achieve the optimal probability. Secondly, we exploit the case where two general pure states are prepared and general measurements are considered. We get an upper bound on the probability for a successful measurement operation to produce a maximally entangled state without any further local operations.

\end{abstract}

\maketitle

\section{Introduction}

In the last three decades, quantum computation and quantum information has become one of the most active research fields. Many significant quantum information processing protocols have been proposed~\cite{NL00}. Remarkable progress has been achieved in both theoretical and experimental aspects. For example, quantum teleportation for sending an unknown quantum bit~\cite{CB93}, quantum key distribution for quantum cryptography~\cite{PS00}, and quantum dense coding for communicating two bits by sending only one qubit which is the inverse of quantum teleportation~\cite{CB92,KM96}. As the most basic and counterintuitive characteristics of quantum mechanics, quantum entanglement plays an indispensable role in all the above applications.

The first step towards the implementation of these applications is to distribute quantum entanglement over remotely located participants. However, there are two bottlenecks for directly sending quantum states over long distance~\cite{NC11}. On one hand, the probability of absorption when transmitting a photon increases exponentially with the distance. For example, the 1 km long fiber has a transmission of $95\%$ while the rate of 1000km fiber is $10^{-10}$ Hz which means transmitting a photon successfully every 300 years.  On the other hand, even when a photon arrives at the destination, the fidelity of the transmitted state decreases exponentially with the distance because of the noise in quantum channels.

A solution for this problem is to use quantum repeaters to divide the long distance into many shorter segments. Each of the segments has tolerable probability for absorption and noise. The first protocol of quantum repeaters was proposed by Briegel~\cite{HW98}, and experimental progresses have been made since then~\cite{ZT03, AW03}.

In the basic scenario of quantum repeater, two copies of a bipartite quantum state $\rho$ are shared by three participants, say Alice, Clare and Bob. Alice and Clare share the copy $\rho_{AC}$, and Clare and Bob share $\rho_{CB}$.
By performing local operations and classical communication (LOCC) between these three parties, quantum entanglement can be created between Alice and Bob, thereby extending the distance of entanglement distribution. The basic scenario is depicted in Fig~\ref{fig:Scenario1}.
\begin{figure}[h]
   \centering
   \includegraphics[width=0.5\textwidth]{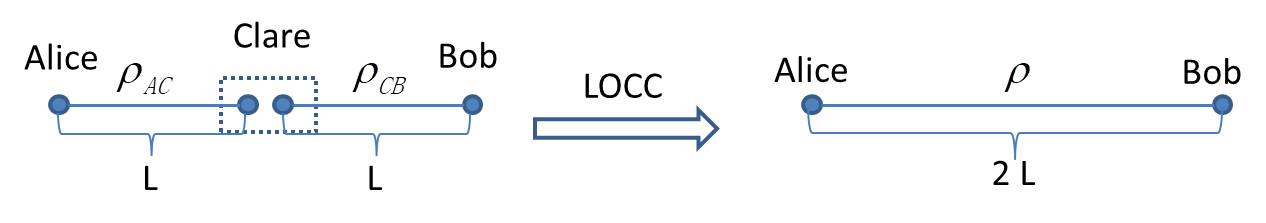}
   \caption{Basic scenario of quantum repeaters.} \label{fig:Scenario1}
\end{figure}

This basic scenario is also known as quantum entanglement swapping~\cite{ZM93}. Bose et al. considered the case where the same two-qubit pure states are prepared in each segments~\cite{SV98}. They found the optimal strategy of quantum swapping in the perspective of concentrating the most entanglement between Alice and Bob. Shi et al. considered the case where different two-qubit pure states are prepared~\cite{BY00}. In their strategy, Clare performs a projective measurement in the standard Bell basis. Then, Alice and Bob perform local operations to create maximally entangled state between them. Shi et al. found out that the optimal entanglement concentration rate is exactly the same as the concentration rate of the scenario where the less entangled one resource is directly distributed between Alice and Bob. Hardy and Song considered the case where general entangled pure states are distributed in a chain scenario~\cite{LD00}. Several groups have experimentally demonstrated entanglement swapping~\cite{CT09, XC16}.

In the basic quantum repeater scheme, the protocol consists of preparation of quantum resources and LOCC. In practice, quantum resources can be efficiently transmitted over a short distance~\cite{NC11}. Once the quantum resources are prepared, entanglement could be concentrated from the scenario by applying LOCC. Therefore, it is of practical significance to consider the LOCC complexity in the protocol.

In this work, we exploit the basic quantum repeater scheme in perspectives of both entanglement concentration rate and LOCC complexity. Firstly, we exploit the scenario where two different two-qubit pure states are prepared. We construct a protocol which can create entanglement between Alice and Bob with the optimal rate and less consumption of LOCC resources. We also find a criteria for the projective measurements which are able to achieve the optimal rate of entanglement concentration. Secondly, we exploit the scenario where two different higher dimensional pure states are prepared and general measurements are considered. We get an upper bound on the probability for a successful measurement operation to produce a maximally entangled state between Alice and Bob without any further local operations.

\section{Taking different two-qubit pure states as resource}
In this section, we consider the case where two different pure states of a two-qubit system are prepared in the basic quantum repeater configuration. We construct a protocol to concentrate maximally entangled two-qubit state between Alice and Bob with the optimal entanglement concentration rate and less LOCC complexity.

Any pure state of a two-qubit system can be expressed as $|\Phi_{\theta}\rangle \equiv \cos\theta |00\rangle + \sin\theta |11\rangle, \theta\in[0,\frac{\pi}{2}]$ up to a local unitary on either party. Note that local unitary operations cannot affect quantum entanglement which is a kind of resource that could not be created or increased with LOCC. Thus, we only need to consider the case that the entangled states $|\Phi_{\theta}\rangle$ and $|\Phi_{\eta}\rangle$ with $\theta, \eta\in (0, \frac{\pi}{4}]$ are initially distributed in the configuration which is depicted in Fig.~\ref{fig:Scenario2QBPure}.
\begin{figure}[h]
   \centering
   \includegraphics[width=0.5\textwidth]{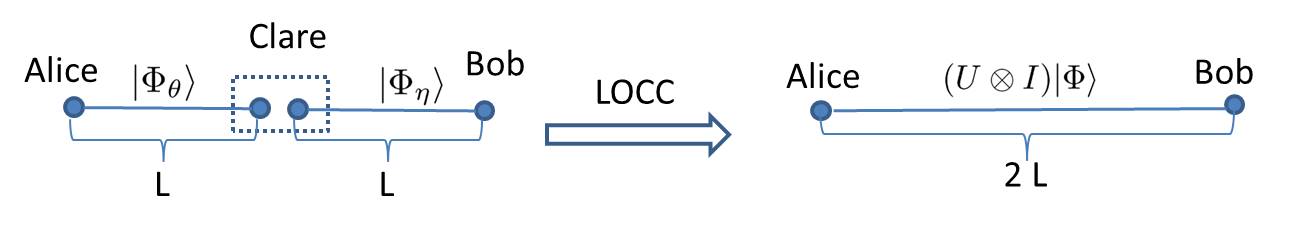}
   \caption{Quantum repeater scheme with different two-qubit pure states.} \label{fig:Scenario2QBPure}
\end{figure}

\begin{lemma}\label{lm:ETtransformProb}
   Suppose two parties share a two-qubit state $|\Phi_{\lambda}\rangle = \cos{\lambda} |00\rangle + \sin{\lambda} |11\rangle$ with $\lambda\in(0, \frac{\pi}{2})$. By performing LOCC, the state can be probabilistically transferred into a maximally entangled state. The probability of the successful transformation is upper bounded by $P_{E}(\Phi_{\lambda}) \equiv \min\{2\cos^{2}\lambda, 2\sin^{2}\lambda\}$.
\end{lemma}
Lemma~\ref{lm:ETtransformProb}, which is an implication of Vidal's result~\cite{GV99}, gives the upper bound on the entanglement concentration rate of the scenario where an entangled two-qubit state $|\Phi_{\lambda}\rangle$ is shared by two parties. With out loss of generality, suppose $\cos{\lambda} \ge \sin{\lambda}$. The upper bound concentration rate could be achieved by performing a general measurement $\{M_{0}, M_{1}\}$, where $M_{0} = \tan{\lambda} |0\rangle\langle 0| + |1\rangle\langle 1|$ and $M_{1} = \sqrt{1 - \tan^{2}{\lambda}}|0\rangle\langle 0|$, on either party. If the measurement outcome $0$ was observed, the maximally entangled state would be created between them. The corresponding probability is $P_{E}(\Phi_{\lambda}) = 2\sin^{2}{\lambda}$.

By applying Lemma~\ref{lm:ETtransformProb}, we get the upper bound on the entanglement concentration rate for the simplest quantum repeater scenario depicted in Fig.~\ref{fig:Scenario2QBPure}, which is concluded in the following lemma.

\begin{lemma}\label{lm:optConcentrationRate}
   Suppose that the two-qubit states $|\Phi_{\theta}\rangle$ and $|\Phi_{\eta}\rangle$ are initially distributed in the scenario depicted in Fig.~\ref{fig:Scenario2QBPure}. Let $P_{MS}$ be the optimal probability that maximally entangled state can created between Alice and Bob by applying LOCC. Then, $P_{MS} \le \min\{ P_{E}(\Phi_{\theta}), P_{E}(\Phi_{\eta}) \}$.
\end{lemma}
\begin{proof}
   We prove this lemma via deducing contradictions. Suppose $P_{MS} > P_{E}(\Phi_{\theta})$. Then, we consider a scenario where the resource $|\Phi_{\theta}\rangle$ is shared by Alice and Bob. Further, let Bob locally prepare an ancilla state $|\Phi_{\eta}\rangle$. According to the definition, the probability of creating maximally entangled state between Alice and Bob by applying LOCC is $P_{MS}$, which is greater than $P_{E}(\Phi_{\theta})$. This result contradicts with Lemma~\ref{lm:ETtransformProb}. Thus, the assumption is not true. It should have $P_{MS} \le P_{E}(\Phi_{\theta})$. Similarly, we can get $P_{MS} \le P_{E}(\Phi_{\eta})$. Therefore, we have $P_{MS} \le \min\{ P_{E}(\Phi_{\theta}), P_{E}(\Phi_{\eta}) \}$.
\end{proof}

In our protocol, we consider a projective measurement on Clare's joint system. A successful projection of Clare is the one when maximally entangled state could be directly created between Alice and Bob without any further local operations. Our strategy is to construct a projective measurement such that the sum of the probabilities of the successful projections is as high as possible. For the projections that result Alice and Bob in less entangled states, we can apply the probabilistic entanglement concentration by performing a local measurement operation on Bob's system. In general, we can concentrate entanglement from the scenario with high probability and less LOCC complexity.

In the following lemma, we work out the lower and upper bounds on the probability of a successful projection.
\begin{lemma}
   Assume that the two-qubit pure states $|\Phi_{\theta}\rangle$ and $|\Phi_{\eta}\rangle$ are initially distributed in the scenario which is showed in Fig.~\ref{fig:Scenario2QBPure}. Without loss of generality, suppose $\theta,\eta\in(0,\frac{\pi}{4}]$. A maximally entangled state could be created between Alice and Bob by projecting Clare's joint system onto the state $|\varphi\rangle$ without any further local operations. The probability $p(\varphi)$ of the projection is bounded by
   \begin{align}\label{projectionBounds}
      \frac{\sin^{2}{2\theta}\sin^{2}{2\eta}} {4(1 + \cos{2\theta} \cos{2\eta})} \le p(\varphi) \le \frac{\sin^{2}{2\theta}\sin^{2}{2\eta}}{4(1 - \cos{2\theta} \cos{2\eta})}.
   \end{align}
\end{lemma}
\begin{proof}
The initial state of the three participants' joint quantum system can be written as
\begin{align}
   |\phi_{0}\rangle_{ABC} = \sum_{k=0}^{3} f_{k}|k\rangle_{AB}|k\rangle_{C}
\end{align}
where $f_{0} = \cos{\theta}\cos{\eta}, f_{1} = \cos{\theta}\sin{\eta}, f_{2} = \sin{\theta}\cos{\eta}$ and $f_{3} = \sin{\theta}\sin{\eta}$. Suppose $|\varphi\rangle = \sum_{k=0}^{3}\mu_{k}|k\rangle \in \mathcal{H}_{2}^{\otimes 2}$ with the constraint $\sum_{k=0}^{3}|\mu_{k}|^{2} = 1$. Once the projection happened, the state of Alice and Bob's joint system would be
\begin{align}
   |\phi\rangle_{AB} = \frac{1}{\sqrt{p(\varphi)}}\langle \phi_{C} |\phi_{0}\rangle_{ABC} = \frac{1}{\sqrt{p(\varphi)}} \sum_{k=0}^{3}f_{k}\mu_{k}^{*}|k\rangle_{AB}. \nonumber
\end{align}
Note that any maximally entangled state of a two-qubit system can be written as $(U \otimes I)|\Phi\rangle$ where $U$ is a unitary on $\mathcal{H}_{2}$ and $|\Phi\rangle \equiv (|00\rangle + |11\rangle)/\sqrt{2}$  is the standard maximally entangled state. As we expect that a maximally entangled state would be created between Alice and Bob without any further local operations, it should also have
\begin{align}
   |\phi\rangle_{AB} = (U \otimes I)|\Phi\rangle
\end{align}
for a unitary $U$.

Let $\Delta_{k} \equiv \langle k|(U \otimes I)|\Phi\rangle$. Then, the parameters of the projection state and the corresponding probability are related by the formula as follows
\begin{align}
   f_{k}\mu_{k}^{*} = \sqrt{p(\varphi)}\Delta_{k}.
\end{align}
By applying the unit constraint of the projection state $|\varphi\rangle$, we get the probability of the successful projection as follows
\begin{align}
   p(\varphi) = (\sum_{k=0}^{3} \frac{|\Delta_{k}|^{2}}{f_{k}^{2}})^{-1}.
\end{align}

Without loss of generality, let $U=|0\rangle\langle\alpha_{0}| + |1\rangle\langle\alpha_{1}|$ where $|\alpha_{0}\rangle = e^{i\tau_{0}}(\cos{\alpha}|0\rangle + e^{i\gamma}\sin{\alpha}|1\rangle)$, $|\alpha_{1}\rangle = e^{i\tau_{1}}(\sin{\alpha}|0\rangle + e^{i\gamma}\cos{\alpha}|1\rangle)$ and $\alpha,\tau_{0},\tau_{1},\gamma\in[0,2\pi)$. It is trivial to find out that $|\Delta_{0}|^{2} =   |\Delta_{3}|^{2} = \frac{1}{2}\cos^{2}{\alpha}$ and $|\Delta_{1}|^{2} =  |\Delta_{2}|^{2} = \frac{1}{2}\sin^{2}{\alpha}$.
Hence, any successful projection $|\varphi\rangle$ should be equivalently written as
\begin{align}\label{eq:generalSuccessProjection}
   |\varphi\rangle = & \sqrt{\frac{p(\varphi)}{2}}(\cos{\alpha}(\frac{1}{f_{0}} |00\rangle + e^{i\beta}\frac{1}{f_{3}} |11\rangle)  \nonumber \\
                     & + e^{i\beta^{'}} \sin{\alpha}(\frac{1}{f_{1}} |01\rangle + e^{i\beta^{''}}\frac{1}{f_{2}} |10\rangle))
\end{align}
where $\alpha,\beta,\beta^{'},\beta^{''}\in[0,2\pi)$. The inverse of the probability can be rewritten as
\begin{align}
   p(\varphi)^{-1} = \frac{1}{2}(\frac{1}{f_{1}^{2}} + \frac{1}{f_{2}^{2}}) + \frac{1}{2}\cos^{2}{\alpha} (\frac{1}{f_{0}^{2}} + \frac{1}{f_{3}^{2}} - \frac{1}{f_{1}^{2}} - \frac{1}{f_{2}^{2}}). \nonumber
\end{align}
As we have assumed that $\theta,\eta\in(0,\frac{\pi}{4}]$, it follows that
\begin{align}
   \frac{1}{f_{0}^{2}} + \frac{1}{f_{3}^{2}} - \frac{1}{f_{1}^{2}} - \frac{1}{f_{2}^{2}} = \frac{16\cos{2\theta}\cos{2\eta}}{\sin^{2}{2\theta}\sin^{2}{2\eta}} \ge &0. \nonumber
\end{align}
Thus, we get
\begin{align}
   p(\varphi)^{-1} & \le \frac{1}{2}(\frac{1}{f_{0}^{2}} + \frac{1}{f_{3}^{2}}) = \frac{4(1 + \cos{2\theta} \cos{2\eta})}{\sin^{2}{2\theta}\sin^{2}{2\eta}}, \\
   p(\varphi)^{-1} & \ge \frac{1}{2}(\frac{1}{f_{1}^{2}} + \frac{1}{f_{2}^{2}}) = \frac{4(1 - \cos{2\theta} \cos{2\eta})}{\sin^{2}{2\theta}\sin^{2}{2\eta}}.
\end{align}
Then, the Eq.~(\ref{projectionBounds}) follows immediately.  Therefore, we have proved the lemma.
\end{proof}

Suppose Clare performs a projective measurement in the orthonormal basis $\{|\varphi_{k}\rangle\}_{k=1}^{4}$ with $p(\varphi_{k})$ being the probability of projecting the system into the state $|\varphi_{k}\rangle$ and $|\phi_{k}\rangle_{AB}$ being the corresponding post-measurement state of Alice and Bob's joint system.

Firstly, let $|\varphi_{1}\rangle$ be the successful projection such that $p(\varphi_{1})$ reaches the upper bound of the projection probability in Eq.~(\ref{projectionBounds}). The condition of achieving the upper bound probability is $\cos\alpha = 0$ in Eq.~(\ref{eq:generalSuccessProjection}). Thus, the projection state can be written as
\begin{equation}\label{maxprobprojection1}
   |\varphi_{1}\rangle = (f_{2}|01\rangle + e^{i\beta_{1}}f_{1}|10\rangle)/\sqrt{f_{1}^{2}+f_{2}^{2}}
\end{equation}
for an arbitrary phase $\beta_{1}\in[0,2\pi)$. The corresponding projection probability is $p(\varphi_{1}) = \frac{\sin^{2}{2\theta}\sin^{2}{2\eta}}{4(1 - \cos{2\theta} \cos{2\eta})}$.

In the next step, we try to construct another successful projection state $|\varphi_{2}\rangle$. The state $|\varphi_{2}\rangle$ should be of the form in Eq.~(\ref{eq:generalSuccessProjection}). The orthogonality of the projection states $|\varphi_{2}\rangle$ and $|\varphi_{1}\rangle$ requires $ \langle \varphi_{1}|\varphi_{2}\rangle = 0$ which is equivalent to the following constraint
\begin{align}
   \sqrt{\frac{p(\varphi)}{2(f_{1}^{2} + f_{2}^{2})}} e^{i\beta^{'}} \sin{\alpha}(\frac{f_{2}}{f_{1}} + e^{i(\beta^{''} - \beta_{1})} \frac{f_{1}}{f_{2}}) = 0. \nonumber
\end{align}
To let the above constraint hold for general resources $|\Phi_{\theta}\rangle$ and $|\Phi_{\eta}\rangle$, it should have $\sin{\alpha} = 0$. Thus, we get the second projection state as follows
\begin{align}\label{maxprobprojection2}
   |\varphi_{2} \rangle =  (f_{3} |00\rangle + e^{i\beta_{2}}f_{0}|11\rangle )/\sqrt{f_{0}^{2} + f_{3}^{2}}
\end{align}
where $\beta_{2}\in[0,2\pi)$ is an arbitrary phase and $p(\varphi_{2}) = \frac{\sin^{2}{2\theta}\sin^{2}{2\eta}}{4(1 + \cos{2\theta} \cos{2\eta})}$.

A simple calculation shows that there is not a third projection state which is of the form in Eq.~(\ref{eq:generalSuccessProjection}) and orthogonal to the projections states $|\varphi_{1} \rangle$ and $|\varphi_{2} \rangle$. The other projection states can be chosen as
\begin{align}
   |\varphi_{3}\rangle = & (f_{1}|01\rangle - e^{i\beta_{1}}f_{2} |10\rangle)/\sqrt{f_{1}^{2} + f_{2}^{2}}, \label{maxprobprojection3} \\
   |\varphi_{4}\rangle = & (f_{0}|00\rangle - e^{i\beta_{2}}f_{3} |11\rangle)/\sqrt{f_{0}^{2} + f_{3}^{2}} \label{maxprobprojection4}
\end{align}
with the projection probabilities being $p(\varphi_{3}) = \frac{f_{1}^{4} + f_{2}^{4}}{f_{1}^{2} + f_{2}^{2}}$ and $p(\varphi_{4}) = \frac{f_{0}^{4} + f_{3}^{4}}{f_{0}^{2} + f_{3}^{2}}$, respectively.
The corresponding post-measurement states of Alice and Bob's joint system are as follows
\begin{align}
   |\phi_{3}\rangle_{AB} = & (f_{1}^{2} |01\rangle - e^{-i\beta_{1}}f_{2}^{2} |10\rangle)/\sqrt{f_{1}^{4} + f_{2}^{4}}, \\
   |\phi_{4}\rangle_{AB} = & (f_{0}^{2} |00\rangle - e^{-i\beta_{2}}f_{3}^{2} |11\rangle)/\sqrt{f_{0}^{4} + f_{3}^{4}}.
\end{align}

We have figure out a projective measurement for Clare. Two of the projections would directly leave Alice and Bob's joint system in maximally entangled states while another two leave them in less entangled states. When the latter measurement outcomes are observed, we can concentrate entanglement from the less entangled states by performing local operations on either Alice or Bob's system.

Let $q$ and $q^{'}$ be the probabilities that we can concentrate maximally entangled state from $|\phi_{3}\rangle$ and $|\phi_{4}\rangle$ by performing the corresponding general measurement on Bob's qubit, respectively. As we have assumed $\theta,\eta\in(0,\frac{\pi}{4}]$, it follows that $f_{0} \ge f_{3}$. With out loss of generality, we suppose $\eta \ge \theta$ which implies $f_{1} \ge f_{2}$. By applying Lemma~\ref{lm:ETtransformProb}, we get
\begin{align}
   q = \frac{2f_{2}^{4}}{f_{1}^{4} + f_{2}^{4}}, q^{'} = \frac{2f_{3}^{4}}{f_{0}^{4} + f_{3}^{4}}. \nonumber
\end{align}
In general, the maximally entangled resource can be created between Alice and Bob with probability
\begin{align}
   P_{MS} = & p(\varphi_{1}) + p(\varphi_{2}) + p(\varphi_{3}) * q + p(\varphi_{4}) *q^{'} = 2\sin^{2}{\theta}.  \nonumber
\end{align}
Note that $2\sin^{2}{\theta}$ also equals to the optimal probability $P_{E}(\Phi_{\theta})$ that the maximally entangled state can be concentrated in the scenario where the resource $|\Phi_{\theta}\rangle$ is directly prepared between Alice and Bob. Thus, we have $P_{MS} = P_{E}(\Phi_{\theta})$.

Similarly, we get $P_{MS} = P_{E}(\Phi_{\eta})$ for the case $\theta \ge \eta$. Thus, the maximally entangled resource can be concentrated from the scenario in Fig.~\ref{fig:Scenario2QBPure} with probability $P_{MS} = \min\{ P_{E}(\Phi_{\theta}), P_{E}(\Phi_{\eta}) \}$. In fact, this is the optimal entanglement concentration rate for the scenario. It is obvious that $P_{MS}$ should be not greater than both $P_{E}(\Phi_{\theta})$ and $ P_{E}(\Phi_{\eta})$. Otherwise, say $P_{MS} > P_{E}(\Phi_{\theta})$, maximally entangled resource would be concentrated from the bipartitely distributed resource $|\Phi_{\theta}\rangle$ with probability $P_{MS}$ by locally preparing an ancilla state $|\Phi_{\eta}\rangle$ and applying above strategy, which contradicts with Lemma~\ref{lm:ETtransformProb}.

By performing the projective measurement on Clare's joint system, no further local operation is needed when any of the two successful projections is observed. Maximally entangled resources will be directly created between Alice and Bob with probability $\frac{\sin^{2}{2\theta}\sin^{2}{2\eta}}{2(1 - \cos^{2}{2\theta} \cos^{2}{2\eta})}$. Assume that the classical communication channel is liable. It is possible to establish an efficient classical communication agreement for telling Bob to take the corresponding action.

Therefore, we have constructed a protocol for extending the distance of entanglement distribution, which is optimal in perspective of entanglement concentration rate and efficient in perspective LOCC complexity. With this, we have proved the following theorem which is the main result of this section.
\begin{theorem}\label{thm:optimalCRofQuantumRepeater}
   Suppose two two-qubit pure states $|\Phi_{\theta}\rangle$ and $|\Phi_{\eta}\rangle$ are initially distributed in the scenario showed in Fig.~\ref{fig:Scenario2QBPure}. Firstly, apply a projective measurement in the orthonormal baisi $\{|\varphi_{k}\}_{k=1}^{4}$, which is defined in Eqs.~(\ref{maxprobprojection1}-\ref{maxprobprojection4}), on Clare's joint system. Secondly, selectively perform local operation on Bob's system according to Clare's measurement outcome. Then, maximally entangled resource could be created between Alice and Bob with probability $P_{MS} = \min\{ P_{E}(\Phi_{\theta}), P_{E}(\Phi_{\eta}) \}$. This protocol is efficient in perspectives of both entanglement concentration rate and LOCC complexity.
\end{theorem}

\subsection{Criteria for optimal projective measurements in perspective of entanglement concentration rate}
For the scenario depicted in Fig.~\ref{fig:Scenario2QBPure}, the projective measurement in the standard Bell basis is also able to achieve the optimal entanglement concentration rate~\cite{BY00}. Thus, the optimal projective measurement for Clare's joint system is not unique in perspective of entanglement concentration rate. However, not all of projective measurements could achieve that goal. For example, it is impossible to obtain entanglement between Alice and Bob by projecting Clare's system into separable states.

Hence, it is a practical problem to verify whether a projective measurements is able to concentrate entanglement from the scenario with the optimal rate. The following theorem gives a criteria for such measurements.
\begin{theorem}
   We consider a scenario with the states $|\Phi_{\theta}\rangle$ and $|\Phi_{\eta}\rangle$, which is showed in Fig.~\ref{fig:Scenario2QBPure}. With out loss of generality, suppose $\theta,\eta\in(0,\frac{\pi}{4}]$ and $\theta\le \eta$. Suppose $P \equiv \{P_{k} \}_{k=1}^{4}$ is a projective measurement on $\mathcal{H}_{2}^{\otimes 2}$ where $P_{k}$ are projectors.  With assistance of LOCC, the projective measurement $P$ on Clare's joint system is able to achieve the optimal entanglement concentration rate if and only if
   \begin{equation}
      \sum_{k=1}^{4} \sqrt{(tr(  T_{1} \otimes T_{2}P_{k}))^{2} + \sin^{2}{2\theta}|tr(|0\rangle \langle 1| \otimes T_{2}P_{k})|^{2}} = \cos{2\theta}
   \end{equation}
   where the operators $T_{1} = \cos^{2}{\theta} |0\rangle \langle 0| - \sin^{2}{\theta} |1\rangle \langle 1|$ and $T_{2} = \cos^{2}{\eta} |0\rangle \langle 0| + \sin^{2}{\eta} |1\rangle \langle 1|$.
\end{theorem}

\begin{proof}
   Suppose $P_{k} = |\varphi_{k}\rangle \langle \varphi_{k}|$ where $\{|\varphi_{k}\rangle\}_{k=1}^{4}$ is an orthonormal basis of the space $\mathcal{H}_{2}^{\otimes 2}$. Note that the initial state of the configuration is $|\Phi_{\theta},\Phi_{\eta}\rangle = \sum_{t=0}^{3}f_{t}|t\rangle_{C}|t\rangle_{AB}$. In the case that the measurement outcome $k$ is observed, the post-measurement state of Alice and Bob's joint system will be
   \begin{equation}
      |\phi_{k}\rangle_{AB} = \frac{1}{\sqrt{p_{k}}} \sum_{t = 0}^{3}f_{t} \langle \varphi_{k} |t\rangle |t\rangle_{AB}  \nonumber
   \end{equation}
   where $p_{k}$ is the probability of observing the measurement result ${k}$. From Lemma~\ref{lm:ETtransformProb}, the probability $P_{E}(\phi_{k})$ is twice of the square of the state's minimal Schmidt number. It is also equal to the twice of the minimal eigenvalue of the density operator of either Alice or Bob's system. On condition that Alice and Bob share the state $|\phi_{k}\rangle_{AB}$, the density operator of Alice's system is
   \begin{align}
      \rho_{A}^{k}  \equiv   tr_{B}(|\phi_{k}\rangle_{AB}\langle \phi_{k}|) = \frac{1}{p_{k}} \sum_{t_{1},t_{2} = 0}^{1}  a_{t_{1}t_{2}}^{k} |t_{1}\rangle \langle t_{2}|.  \nonumber
   \end{align}
   where $a_{t_{1}t_{2}}^{k} = \theta_{t_{1}}\theta_{t_{2}} tr((|t_{1}\rangle \langle t_{2}| \otimes T_{2}) P_{k})$. For the purpose of convenience, we take the notation $\theta_{0} \equiv \cos{\theta}$ and $\theta_{1} \equiv \sin{\theta}$. The eigenvalues of $\rho_{A}^{k}$ are
   \begin{align}
      \lambda_{\pm}^{k} =  \frac{1}{2p_{k}}((a_{00}^{k} + a_{11}^{k}) \pm \sqrt{(a_{00}^{k} - a_{11}^{k})^{2} + 4 a_{01}^{k} a_{10}^{k}}).  \nonumber
   \end{align}
   As $\lambda_{+}^{k} \ge \lambda_{-}^{k}$, we get $P_{E}(\phi_{k}) = 2\lambda_{-}^{k}$.

   When all the measurement outcomes are considered, maximally entangled state could be concentrated from the scenario with rate
   \begin{align}
      p_{s} \equiv & \sum_{k = 1}^{4}  p_{k}\cdot P_{E}(\phi_{k}) \nonumber \\
       = & \sum_{k = 1}^{4}((a_{00}^{k} + a_{11}^{k}) - \sqrt{(a_{00}^{k} - a_{11}^{k})^{2} + 4 a_{01}^{k} a_{10}^{k}}).  \nonumber
   \end{align}
   Note that $\sum_{k = 1}^{4}a_{00}^{k} = \sum_{k = 1}^{4}\cos^{2}\theta tr( ( |0\rangle \langle 0| \otimes T_{2})P_{k}) = \cos^{2}\theta$. Similarly, we get $\sum_{k = 1}^{4}a_{11}^{k} = \sin^{2}\theta$. It is trivial to see that $a_{00}^{k} - a_{11}^{k} = tr((T_{1} \otimes T_{2})P_{k})$ and $4a_{01}^{k}a_{10}^{k} = \sin^{2}{2\theta}|tr((|0\rangle\langle 1| \otimes T_{2})P_{k})|^{2}$. Thus, the rate can be equivalently written as
   \begin{align}
         p_{s} = 1 -  \sum_{k=1}^{4} &  \sqrt{(tr(  (T_{1} \otimes T_{2})P_{k}))^{2} + } \nonumber \\
         & \overline{ \sin^{2}{2\theta} |tr((|0\rangle\langle 1| \otimes T_{2}) P_{k}|^{2}}. \nonumber
   \end{align}

   According to Theorem~\ref{thm:optimalCRofQuantumRepeater}, the optimal entanglement concentration rate of the scenario showed in Fig.~\ref{fig:Scenario2QBPure} is $2\sin^2{\theta}$. Therefore, the projective measurement $P$ is able to achieve the optimal entanglement concentration rate, which means $p_{s} = 2\sin^2{\theta}$, if and only if
   \begin{align}
       \sum_{k=1}^{4} \sqrt{(tr(  T_{1} \otimes T_{2}P_{k}))^{2} + \sin^{2}{2\theta}|tr(|0\rangle \langle 1| \otimes T_{2}P_{k})|^{2}} = \cos{2\theta}.  \nonumber
   \end{align}
\end{proof}

With simple calculations, it shows that both the projective measurement in standard Bell basis and the one we proposed in this section fulfill the criteria.

\section{The upper bound for the scenario with different general pure states}
In this section, we exploit a more general scenario where different general pure states are prepared and general measurements are consider. Suppose a general bipartite pure state $|\psi_{AB}\rangle$ is shared by Alice and Clare while $|\psi_{CB}\rangle$ is shared by Clare and Bob. We analyze the measurement outcome which leaves Alice and Bob in a maximally entangled state without any further local operations. We refer such measurement outcome as a successful measurement outcome. The scenario is depicted in Fig.~\ref{fig:Senario_General_Pure}. Note that an arbitrary maximally entangled state can be written as $|\Omega_{U}\rangle \equiv (U\otimes I)|\Omega\rangle$, where $U$ is a local unitary on $\mathcal{H}_{d}$ and $|\Omega\rangle = \frac{1}{\sqrt{d}} \sum_{k=1}^{d} |k\rangle |k\rangle$ is the standard maximally entangled state.
\begin{figure}[h]
   \centering
   \includegraphics[width=0.5\textwidth]{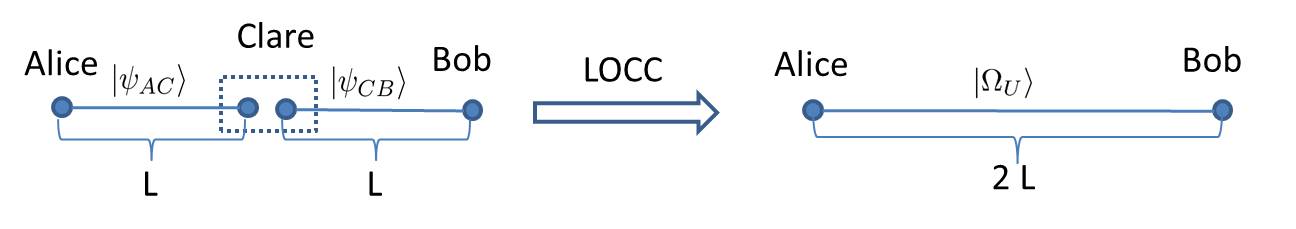}
   \caption{Quantum repeater scenario with two general pure states.}\label{fig:Senario_General_Pure}
\end{figure}

The probability of observing a successful measurement outcome is variant for different $U$. Theorem~\ref{thm:UB4generalpurestates} gives an upper bound on the probability, which is the main result of this section. To prove Theorem~\ref{thm:UB4generalpurestates}, we need two lemmas. Lemma~\ref{lm:QuantumSteering} is concluded from Wolf's lecture notes~\cite{MW12}. Lemma~\ref{lm:lowerboundtraceOperatorMulti} is a generalization of a mathematical theorem which is proved by Lewis~\cite{AS96}. We denote $\lambda_{\downarrow}(A)$ for the collum vector of operator $A$'s eigenvalues in the nonincreasing order and $\lambda_{\uparrow}(A)$ for that in the nondecreasing order.
\begin{lemma}\label{lm:QuantumSteering}
   Suppose $|\psi\rangle$ is a bipartite pure state of the joint quantum system $A\otimes B$. Let $\rho_{A}\equiv tr_{B}(|\psi\rangle \langle \psi|)$ be the density operator of the subsystem $A$. Suppose $\rho_{A}$ could be expressed as a convex combination $\rho_{A} = \sum_{i}\lambda_{i}\rho_{i}$ where $\lambda_{i}>0$, $\sum_{i} \lambda_{i} = 1$ and $\rho_{i}$ are density operators on $\mathcal{H}_{A}$. Then, there is a quantum measurement operation on system $B$, say $T=\{ T_{i}: \mathcal{B(H_{B})} \rightarrow \mathcal{B(H_{B})} \}$, such that
   \begin{align}
      \lambda_{i}\rho_{i} = tr_{B}[(I\otimes T_{i})(|\psi\rangle \langle \psi |)].
   \end{align}
   The parameter $\lambda_{i}$ can be interpreted as the probability of observing the measurement outcome $i$. The  upper bound of $\lambda_{i}$ is as follows
   \begin{align}
      \lambda_{i} \le \parallel \rho^{-\frac{1}{2}}\rho_{i} \rho^{-\frac{1}{2}} \parallel^{-1}_{\infty}.
   \end{align}
\end{lemma}

\begin{lemma}\label{lm:lowerboundtraceOperatorMulti}
   For Hermitian operators $A$ and $B$,
   \begin{align}\label{ieq:LBtraceOperatorMulti}
      tr(AB) \ge \lambda_{\uparrow}(A)^{T}\lambda_{\downarrow}(B)
   \end{align}
   with equality if and only if there is a unitary operator $U$ such that $U^{\dagger}AU = diag(\lambda_{\uparrow}(A))$ and $U^{\dagger}BU = diag(\lambda_{\downarrow}(B))$.
\end{lemma}

\begin{theorem}\label{thm:UB4generalpurestates}
   Suppose two general pure states $|\psi_{AC}\rangle = \sum_{k=1}^{d_{A}} \sqrt{a_{k}}|k\rangle|k\rangle$ and $|\psi_{CB}\rangle = \sum_{k=1}^{d_{B}} \sqrt{b_{k}}|k\rangle|k\rangle$ are prepared in the scenario, which is showed in Fig.~\ref{fig:Senario_General_Pure}. With out loss of generality, suppose $a_{1}\ge a_{2} \ge \cdots a_{d_{A}} > 0$, $b_{1}\ge b_{2} \ge \cdots b_{d_{B}} > 0$ and $d \equiv d_{B} \ge d_{A}$.  Clare performs a general measurement $M\equiv \{M_{i} \}$ on his joint system with $p_{i}$ being the  probability of observing the measurement outcome $i$. Suppose Alice and Bob would share a maximally entangled state $|\Omega_{U}\rangle \in \mathcal{H}_{d}^{\otimes 2}$ when the measurement outcome $i$ is observed by Clare. Then, the corresponding probability $p_{i}$ has the following upper bound
   \begin{align}
      p_{i} \le \frac{d}{\sum_{k=1}^{d_{A}}\frac{1}{a_{k} b_{d_{A} + 1 - k}}}\equiv p_{max}.
   \end{align}
   The upper bound can be achieved by setting the measurement operator $M_{i} = \sqrt{p_{max}} |\Omega_{U}\rangle \langle \Omega_{U}| \rho_{AB}^{-\frac{1}{2}} $ where $U = \sum_{k= 1}^{d_{A}}|d_{A} + 1 - k \rangle \langle k| + \sum_{k=d_{A}+1}^{d} |k\rangle \langle k|$ and $\rho_{AB}$ is the initial state of the Alice and Bob's joint system.
\end{theorem}
\begin{proof}
   Initially, Alice and Bob's joint system is in the state $\rho_{AB} = \rho_{A} \otimes \rho_{B}$ where $\rho_{A} = \sum_{k=1}^{d} a_{k}|k\rangle \langle k|$ and $\rho_{B} = \sum_{k=1}^{d} b_{k}|k\rangle \langle k|$ . Note that we extend $\rho_{A}$ onto the space $\mathcal{H}_{d}$ by setting $a_{k}=0$ for $k>d_{A}$. In the following discussion, we denote $a_{k}^{-1} = 0$ for $k$ such that $a_{k} = 0$.

   According to Lemma ~\ref{lm:QuantumSteering}, the probability $p_{i}$ is upper bounded by
   \begin{align}\label{ineq:UBProbNorm}
      p_{i} \le \parallel \rho_{AB}^{-\frac{1}{2}} \rho_{U} \rho_{AB}^{-\frac{1}{2}} \parallel^{-1}_{\infty}.
   \end{align}
   where the inverse operator is defined on the corresponding support space. The equality in Eq.~(\ref{ineq:UBProbNorm}) holds when $M_{i}^{T} = \sqrt{p_{i}} \rho_{AB}^{-\frac{1}{2}} \rho_{U}$. A simple calculation shows that $\rho_{AB}^{-\frac{1}{2}}|\Omega_{U}\rangle = \frac{1}{\sqrt{d}} \sum_{k,t =1}^{d}a_{k}^{-\frac{1}{2}} b_{t}^{-\frac{1}{2}} \langle k|U|t\rangle |k\rangle |t\rangle$. Thus, we get
   \begin{align}
     \| \rho_{AB}^{-\frac{1}{2}} \rho_{U} \rho_{AB}^{-\frac{1}{2}}  \|_{\infty} = & \|\rho_{AB}^{-\frac{1}{2}}|\Omega_{U}\rangle \|^{2} \nonumber \\
     = & \frac{1}{d} \sum_{k,t=1}^{d} a_{k}^{-1}b_{t}^{-1}  \langle k|U|t\rangle \langle t|U^{\dagger}|k\rangle  \nonumber \\
     = & \frac{1}{d} tr(U\rho_{B}^{-1}U^{\dagger} \rho_{A}^{-1}).
   \end{align}
   The eigenvalues of $\rho_{A}^{-1}$ are $a_{d_{A}}^{-1} \ge \cdots \ge a_{1}^{-1} > a_{d_{A}+1}^{-1} = \cdots = a_{d}^{-1} = 0$. That of $U\rho_{B}^{-1}U^{\dagger}$ are $b_{d}^{-1} \ge \cdots \ge b_{1}^{-1} > 0$. By applying Lemma \ref{lm:lowerboundtraceOperatorMulti}, we get
   \begin{align}
      tr(U\rho_{B}^{-1}U^{\dagger} \rho_{A}^{-1}) \ge \sum_{k=1}^{d_{A}}\frac{1}{a_{k} b_{d_{A} + 1 - k}} \nonumber
   \end{align}
   where the equality holds when
   \begin{align}\label{eq:UPconditionU}
      U = (\sum_{k=1}^{d_{A}}|d_{A} + 1 -k\rangle \langle k|) \oplus I
   \end{align}
   where the term in the direct sum acts on the kernel space of $\rho_{A}$.

   Therefore, we get the upper bound of the the probability $ p_{i}$ as follows
   \begin{align}
      p_{i} \le \frac{d}{\sum_{k=1}^{d_{A}}\frac{1}{a_{k} b_{d_{A} + 1 - k}}}
   \end{align}
   where the equality holds when the measurement operator is $M_{i} = \sqrt{p_{max}} |\Omega_{U}\rangle \langle \Omega_{U}| \rho_{AB}^{-\frac{1}{2}} $ and $U$ takes the form defined in Eq.~(\ref{eq:UPconditionU}).
\end{proof}

\section{Conclusion and discussion}
In this paper, we have exploited the basic configuration of quantum repeaters in perspectives of both entanglement concentration rate and LOCC complexity. For the scenario with two different two-qubit pure states, we constructed a protocol to concentrate entanglement. The protocol is optimal in perspective entanglement concentration rate and efficient in perspective LOCC complexity. We also find a criteria for the projective measurement to achieve the optimal entanglement concentration rate. For the scenario where general pure states are prepared and general measurements are considered, we get the upper bound on the probability of a successful measurement outcome which produces a maximally entangled state between Alice and Bob without any further local operations.

The protocol is composed of two steps. Firstly, Clare performs a measurement operation. Secondly, based on Clare's measurement outcome, Bob chooses the corresponding strategy; namely, Bob does not do any further local operations or performs the corresponding general measurements. In general, the protocol can concentrate entanglement from the scenario with the optimal rate. We reduced the LOCC complexity via the strategy that no further local operation was needed if maximally entangled state could be created between Alice and Bob when a measurement outcome was observed by Clare.

For the scenario with different two-qubit states $|\Phi_{\theta}\rangle$ and $|\Phi_{\eta}\rangle$, we constructed a projective measurement such that maximally entangled state could be created after two of the four measurement outcomes. Such measurement outcomes could be observed with probability $\frac{\sin^{2}{2\theta} \sin^{2}{2\eta}} {2(1 - \cos^{2}{2\theta} \cos^{2}{2\eta})}$. If the states for the scenario are same, say $|\Phi_{\theta}\rangle$, we can construct a projective measurement with three successful projections to produce maximally entangled states without any further local operations. The corresponding probability is $\frac{\sin^{2}{2\theta}(3 + \cos^{2}{2\theta})} {4(1 + \cos^{2}{2\theta})}$.

\begin{acknowledgments}
    The author is delighted to thank Professor Yuan Feng for illuminating and fruitful discussions. This research is partially supported by Chinese Scholarship Council (Grant No. 201206270069), Australian Research Council (Grant No. DP160101652) and National Natural Science Foundation of China (Grant No. 61472452 and 61772565).
\end{acknowledgments}



\end{document}